\documentclass[12pt]{amsart}

\usepackage{amsfonts}
\usepackage{epsfig}
\usepackage{graphicx}
\usepackage{amsmath}
\usepackage{amssymb}
\usepackage{mathrsfs}

\DeclareMathAlphabet{\mathpzc}{OT1}{pzc}{m}{it}

\newcounter{main}

\newtheorem{theorem}{Theorem}[section]

\newtheorem{lemma}[theorem]{Lemma}

\newtheorem{definition}{Definition}[section]
\newtheorem{maintheorem}{Theorem}

\newcommand{\blanksquare}{\,\,\,$\sqcup\!\!\!\!\sqcap$}

\newcounter{example}
\newenvironment{example}%
{{\stepcounter{example}}{\flushleft {\bf Example \arabic{example}:}}}%
{\par}

\title[Perturbations of Mathieu equations]
{Perturbations of Mathieu equations with parametric excitation of large period}

\date{22 June 2010}

\author[M. Bessa]{M\'{a}rio Bessa}
\address{Departamento de Matem\'atica da Universidade do Porto, 
Rua do Campo Alegre, 687, 
4169-007 Porto, Portugal \\ ESTGOH-Instituto Po\-li\-t\'ec\-ni\-co de Coimbra, Rua General Santos Costa, 3400-124 Oliveira do Hospital, Portugal}
\email{bessa@fc.up.pt}

\begin{document}

\begin{abstract}
We consider a linear differential system of Mathieu equations with periodic coefficients over periodic closed orbits and we prove that, arbitrarily close to this system, there is a linear differential system of Hamiltonian damped Mathieu equations with periodic coefficients over periodic closed orbits such that, all but a finite number of closed periodic coefficients, have unstable solutions. The perturbations will be peformed in the periodic coefficients.
\end{abstract}

\maketitle

\emph{keywords:} Mathieu equation, characteristic multipliers, linear differential systems.\\

\begin{section}{Introduction}

\begin{subsection}{The Mathieu equation}

Let $\psi\colon \mathbb{R}\rightarrow \mathbb{R}$ be a differentiable $\tau$-periodic function, i.e., $\psi(t+\tau)=\psi(t)$ for some $\tau>0$ and all $t\in\mathbb{R}$. The \emph{Hill's equation} is defined by 
\begin{equation}\label{HE}
\ddot{y}+\psi(t)y=0,\,\, t\geq 0.
\end{equation}
Fixing the parameters $\omega, \epsilon>0$ and letting $\psi(t)=\cos(2\pi t)$ we obtain a very special case of a Hill equation - the \emph{Mathieu equation}: 
\begin{equation}\label{ME0}
\ddot{y}+(\omega^2+\epsilon \cos(2\pi t))y=0,\,\, t\geq 0.
\end{equation}
These type of second-order differential equations, were introduced by Mathieu~\cite{Ma} in the study of oscillations. The parameters $\omega^2$ and $\epsilon\cos(2\pi t)$ represent respectively, the frequency of the oscillation and the parametric excitation of strength $\epsilon$. 

In what follows we generalize and assume that the parametric excitation of strength $\epsilon$ can take the form $\epsilon\psi(t)$, where $\psi(t)$ is any function in the broader setting of smooth periodic functions. So, from now on we call the differential equation of second order
\begin{equation}\label{ME}
\ddot{y}+(\omega^2+\epsilon \psi(t))y=0,\,\, t\geq 0,
\end{equation}
a Mathieu equation. 

More generally, by inputing a first order term, we get the \emph{damped Mathieu equation} (see~\cite{TN}) which is defined by 
\begin{equation}\label{DME}
\ddot{y}+\gamma \dot{y}+(\omega^2+\epsilon \psi(t))y=0,\,\, t\geq 0, 
\end{equation}
where $\gamma>0$ is very small and, for this reason, (\ref{DME}) can be seen as a small perturbation of (\ref{ME}). In this paper we consider damped Mathieu equations with periodic coefficients. In fact, we say that
\begin{equation}\label{DMEP}
\ddot{y}+\Gamma(t)\dot{y}+(\omega^2+\epsilon \Psi(t))y=0,\,\, t\geq 0, 
\end{equation}
is a \emph{damped Mathieu equation with periodic coefficients $\Gamma(t)$ and $\Psi(t)$} if $\Gamma(t)$ and $\Psi(t)$ are two $\tau$-periodic smooth real functions with $\Gamma(t)$ very close to zero. In the sequel we will consider perturbations of the equation (\ref{ME}) and these perturbations are defined by damped Mathieu equations with periodic coefficients $\Gamma(t)$ and $\Psi(t)$, where $\Gamma(t)$, as we mentioned above, is close to zero and $\Psi(t)$ is close to $\psi(t)$. Actually, the central question of this paper is to fix an equation of type (\ref{ME}) and then, keeping the parameters $\omega$ and $\epsilon$ fixed and allowing some perturbations of the periodic coefficients $\Gamma(t)$ and $\Psi(t)$, understand the typical dynamical behavior of its solutions. To be more precise we will make use of an abstract and very general result about the $C^0$ perturbations of the dynamical cocycle along closed orbits (see~\cite{BGV,BR}) and connect it with the second-order differential equations of type (\ref{ME}) and (\ref{DMEP}) via its representation as linear differential systems.

\end{subsection}

\begin{subsection}{Linear differential systems}

In order to study families of linear nonau\-to\-no\-mous differential equations like equations (\ref{ME}), (\ref{DME}) and (\ref{DMEP}) it is natural to write them in the terminology of linear differential systems; let $M$ be a compact, boundaryless Hausdorff manifold with dimension $n\geq 2$ and $\varphi^{t}\colon M\rightarrow{M}$ be a flow. A two-dimensional cocycle based on $\varphi^{t}$ is defined by a flow
$\Phi^{t}(p)$ differentiable on the time parameter
$t\in{\mathbb{R}}$,
acting on the general linear group ${{gl}(2,\mathbb{R})}$ (i.e., the set of $2\times 2$ matrices with real entries and nonzero determinant). Together they form the linear
skew-product flow:

$$
\begin{array}{cccc}
\Upsilon^{t}\colon & M\times{\mathbb{R}^{2}} & \longrightarrow & M\times{\mathbb{R}^{2}} \\
& (p,v) & \longmapsto & (\varphi^{t}(p),\Phi^{t}(p)\cdot{v})
\end{array}
$$

The linear flow $\Phi^{t}$ satisfy the so-called \emph{cocycle identity}, that is, for all $t,s\in\mathbb{R}$ and $p\in M$,
$$\Phi^{t+s}(p)=\Phi^{s}(\varphi^{t}(p))\cdot{\Phi^{t}(p)}.$$

Let $\mathfrak{gl}(2,\mathbb{R})$ denote the Lie algebra associated to the Lie group $gl(2,\mathbb{R})$. If we define a map $A\colon M\rightarrow{\mathfrak{gl}(2,\mathbb{R})}$ in
a point $p\in{M}$ by:
$$A(p)=\dot{\Phi}^{s}(p)|_{s=0}$$
and along the orbit $\varphi^{t}(p)$ by:
\begin{equation}\label{lvi}
A(\varphi^{t}(p))=\dot{\Phi}^{s}(p)|_{s=t}\cdot{[\Phi^{t}(p)]^{-1}},
\end{equation}
then $\Phi^{t}(p)$ will be the solution of the \emph{linear variational equation}:
\begin{equation}\label{lve}
\dot{u}(s)_{|s=t}=A(\varphi^{t}(p))\cdot u(t),
\end{equation}
and $\Phi^{t}(p)$ is called the \emph{fundamental matrix} or \emph{fundamental solution}. Given a
cocycle $\Phi^{t}$ we can induce the associated infinitesimal generator $A$ by
using~(\ref{lvi}) and given $A$ we can recover the cocycle by
solving the linear variational equation~(\ref{lve}), from which we get
$\Phi_{A}^{t}$. A two-dimensional \emph{linear differential system} is a three-tuple $(M, \varphi^t, A)$; $\varphi^t$ defines the action in the \emph{base} and $A$ the action in the \emph{fiber}.

In what follows we will be interested in cocycles evolving on the linear groups $sl(2,\mathbb{R})$ and $gl(2,\mathbb{R})$. Recall that $sl(2,\mathbb{R})$ denotes the \emph{special linear group} of matrices with real entries and with determinant equal to one, i.e., they are \emph{area-preserving} or \emph{Hamiltonian}. 

Using Liouville's formula we can relate the trace of $A\in \mathfrak{gl}(2,\mathbb{R})$ with the determinant of $\Phi^t_A$ by:
\begin{equation}\label{liouville}
exp\left({\int_{0}^{t}\text{Tr}A(\varphi^{s}(p))ds}\right)=\det\Phi^{t}(p), 
\end{equation}
where $\text{Tr}(A)$ denoted the trace of $A$.

If $\Phi^t_A$ evolves on $sl(2,\mathbb{R})$, then the associated infinitesimal generator $A$ is \emph{traceless} (cf. (\ref{liouville})) i.e., $\text{Tr}(A)=0$. 

We topologize the space of linear differential systems in the following way; given two functions $A,B\colon M\rightarrow \mathfrak{gl}(2,\mathbb{R})$ over $\varphi^t\colon M\rightarrow M$ we can measure its distance by $$\underset{p\in M}{\max}\|A(p)-B(p)\|,$$ where $\|\cdot\|$ denotes the standard uniform norm.

If the function $A\colon M\rightarrow \mathfrak{gl}(2,\mathbb{R})$ varies continuously with the space variable $p$ (for all $p\in M$) we say that the linear differential system is \emph{continuous}. If the function $A\colon M\rightarrow \mathfrak{gl}(2,\mathbb{R})$ is bounded with respect to the space variable $p$ (for all $p\in M$) we say that the linear differential systems is \emph{bounded}. In this paper we deal with linear differential systems with Mathieu equations in the fiber and also its perturbations (cf. \S \ref{perturb}), the main results in this paper are true for both continuous and bounded settings.

\end{subsection}

\begin{subsection}{Writing equations in the language of linear differential systems}

Given a flow $\varphi^t\colon M\rightarrow M$ we will denote by $\text{Per}(\varphi^t)$ the set of its \emph{closed orbits}. Let $\psi(t)$ be a $\tau$-periodic differentiable coefficient and $p\in \text{Per}(\varphi^t)$ be a closed orbit of period $\pi(p)=\tau$. As a standard reduction on the order of the differential equation via the \emph{momentum} $\dot{y}=z$ we write 
$$\left\{\begin{array}{ccc}
\dot{y}=z\,\,\,\,\,\,\,\,\,\,\,\,\,\,\,\,\,\,\,\,\,\,\,\,\,\,\,\,\,\,\,\,\,\,\,\,  \\
\dot{z}=-(\omega^2+\epsilon\psi(t))y
\end{array}\right.$$
then we obtain

\begin{equation}\label{ME2}
\dot{x}= A(t)\cdot x, 
\end{equation}
where $\dot{x}=(\dot{y},\dot{z})$ and $A(t)=\begin{pmatrix}
0 & 1 \\ -\omega^2-\epsilon\psi(t) & 0
\end{pmatrix}$. 

Once we have the linear variational equation (\ref{ME2}) we are interested in the behavior of the linear differential system induced by the traceless system $A(t)$, i.e., the dynamics of the fundamental matrix solution $\Phi^t_A(p)$.

The eigenvalues of $\Phi_A^{\pi(p)}(p)$ are called the \emph{characteristic multipliers} of the system $A$. It is trivial to see that the product of the characteristic multipliers (or the characteristic multiplier if there is only one) of the system (\ref{ME2}) is equal to one. If they are nonreal conjugates the solution is called \emph{stable} (and the point $p$ \emph{elliptic}), and if one of them is larger than one (in modulus) the solution is called \emph{unstable} (and the point $p$ \emph{hyperbolic}). 

The regions of stability \emph{versus} instability depending on the two parameters of the equation (\ref{ME0}) are well-known (see e.g., \cite[Figure 8.13]{HK}). The tongues are formed by \emph{parabolic} points for $\Phi_A^{\pi(p)}(p)$ (corresponding to characteristic multipliers equal to $1$ or $-1$).

\bigskip

More generally, given any $A(t)=\begin{pmatrix}
\alpha(t) & \beta(t) \\ \gamma(t) & \zeta(t)
\end{pmatrix}\in\mathfrak{gl}(2,\mathbb{R})$, where $\alpha(t)$, $\beta(t)$, $\gamma(t)$ and $\zeta(t)$ are smooth $\tau$-periodic coefficients, the second-order differential equation associated to it is equal to
\begin{equation}\label{form}
\ddot{y}-\left(\text{Tr}(A)+\frac{\dot{\beta}}{\beta}\right)\dot{y}+\left(\det(A)-\dot{\alpha}+\frac{\dot{\beta}\alpha}{\beta}\right)y=0.
\end{equation}

Let us see how we derived (\ref{form}); we have
$$\left\{\begin{array}{ccc}
\dot{y}=\alpha(t)y+\beta(t)z \\
\dot{z}=\gamma(t)y+\zeta(t)z 
\end{array}\right.,$$
now, taking derivatives with respect to (w.r.t.) time, we obtain
\begin{eqnarray*}
 \ddot{y}&=&\dot{\alpha}y+\alpha\dot{y}+\dot{\beta}z+\beta\dot{z}=
\dot{\alpha}y+\alpha\dot{y}+\dot{\beta}\left( \frac{\dot{y}-\alpha y}{\beta}\right)+\beta(\gamma y+\zeta z )\\
&=& \left(\alpha+\frac{\dot{\beta}}{\beta}\right)\dot{y}+\left(\dot{\alpha}-\frac{\dot{\beta}\alpha}{\beta}+\beta\gamma\right)y+\beta\zeta\left( \frac{\dot{y}-\alpha y}{\beta}\right)\\
&=& \left(\alpha+\frac{\dot{\beta}}{\beta}+\zeta\right)\dot{y}+\left(\dot{\alpha}-\frac{\dot{\beta}\alpha}{\beta}+\beta\gamma-\zeta\alpha\right)y\\
&=& \left(\text{Tr}(A)+\frac{\dot{\beta}}{\beta}\right)\dot{y}+\left(-\det(A)+\dot{\alpha}-\frac{\dot{\beta}\alpha}{\beta}\right)y.
\end{eqnarray*}

One trivial observation can be taken from equation (\ref{form}); if $\alpha(t)$ and  $\beta(t)$ are constant, thus $\dot{y}=a y +b z$, for $a,b\in\mathbb{R}$, then (\ref{form}) is reduced to the simple form  $\ddot{y}-\text{Tr}(A)\dot{y}+\det(A)y=0$. For example, we can easily check that equation (\ref{ME}) has this form if we take $\zeta(t)=-a$ and $\gamma(t)=b^{-1}(-\omega^2-\epsilon\psi(t)-a^2)$. 

If we consider a traceless system $A(t)$, then $\alpha(t)+\zeta(t)=0$ and so (\ref{form}) become:

\begin{equation}\label{form2}
\ddot{y}-\frac{\dot{\beta}}{\beta}\dot{y}+\left(\det(A)-\dot{\alpha}+\frac{\dot{\beta}\alpha}{\beta}\right)y=0.
\end{equation}

If, in equation (\ref{DMEP}), we let $\dot{y}=z$ then, $$\left\{\begin{array}{ccc}
\dot{y}=z\,\,\,\,\,\,\,\,\,\,\,\,\,\,\,\,\,\,\,\,\,\,\,\,\,\,\,\,\,\,\,\,\,\,\,\,\,\,\,\,\,\,\,\,\,\,\,\,\,\,\,\,\,\,\,\,   \\
\dot{z}=-(\omega^2+\epsilon\Psi(t))y-\Gamma(t)z
\end{array}\right.$$
and so $\dot{x}= C(t)\cdot x$, where $\dot{x}=(\dot{y},\dot{z})$ and the nontraceless system $$C(t)=\begin{pmatrix}
0 & 1 \\ -\omega^2-\epsilon\Psi(t) & -\Gamma(t)
\end{pmatrix}.$$

On the other hand if, in equation (\ref{ME}), we let $\dot{y}=\alpha(t)y+b\,z$ then, using (\ref{form2}) and forcing the traceless condition $\zeta(t)=-\alpha(t)$, we obtain the traceless linear differential system $\mathcal{A}(t)=\begin{pmatrix}
\alpha(t) & b \\ \gamma(t) & -\alpha(t)
\end{pmatrix}$, where $b\in\mathbb{R}$. Now, since we have two-degrees of freedom, given any smooth $\tau$-periodic function $\alpha(t)$, we define the $\tau$-periodic function:
\begin{equation}\label{form3}
\gamma(t):=b^{-1}(-\omega^2-\epsilon\psi(t)-\dot{\alpha}(t)-\alpha^2(t)).
\end{equation}
Observe that, if we consider $\dot{y}=\alpha(t)y+b\,z$ such that $\alpha(t)$ is $C^1$-close\footnote{The $\delta$-$C^1$-closeness of two functions $f(t)$ and $g(t)$ means that $|f(t)-g(t)|<\delta$ and $|\dot{f}(t)-\dot{g}(t)|<\delta$ for all $t$.} to the zero function and $b$ is close to $1$, then $\mathcal{A}(t)$ can be seen as a small perturbation of $A(t)$, defined in (\ref{ME2}), and both are traceless. If fact, what we allow here is a small perturbation of the momentum $\dot{y}=z$.

Taking this last paragraph into account, we proceed in the same way in equation (\ref{DMEP}) and we let $\dot{y}=\alpha(t)y+\beta(t)z$ where $\alpha(t)$ is $C^1$-close to the zero function and $\beta(t)=e^{-\int_0^t \Gamma(s)ds}$. Finally, we define the $\tau$-periodic function $\gamma(t)$ by:
$$
\gamma(t):=\beta^{-1}(t)\left(-\omega^2-\epsilon\psi(t)-\dot{\alpha}(t)-\alpha^2(t)+\frac{\dot{\beta}(t)\alpha(t)}{\beta(t)}\right).
$$
With this choice for $\dot{y}$ we obtain a \textbf{traceless} linear differential system induced by equation (\ref{DMEP}) very close to $C(t)$.

\end{subsection}

\begin{subsection}{On the neighborhoods of the equation (\ref{ME})}\label{perturb}
In this work we are going to perturb the equation (\ref{ME}) allowing that these perturbations being of the form (\ref{DMEP}). Moreover, we assure that, for a certain perturbation of the momentum, $\dot{y}=\alpha(t)y+\beta(t)z$, the equation (\ref{DMEP}) has a traceless linear differential system associated. 

Fix positive real numbers $\omega,\epsilon$ and $\delta$ and a $\tau$-periodic smooth function $\psi(t)$. This elements give rise to an equation $\mathcal{M}=\mathcal{M}(\omega,\epsilon,\psi(t))$ of the form (\ref{ME}). Let $A$ be the  linear differential system induced by equation (\ref{ME}) with $\dot{y}=z$. We say that $\mathcal{U}(\mathcal{M},\delta)$ is a $\delta$-\emph{neighborhood of $\mathcal{M}$} if:
\begin{enumerate}
\item [(i)] there exists a choice for $\dot{y}$; $\dot{y}=\alpha(t)y+\beta(t)z$, where $\alpha(t)$ is $C^1$-close to the zero function and $\beta(t)$ is $C^1$-close to the function $f(t)=1$;
\item [(ii)] there exists a $\tau$-periodic coefficient $$\gamma(t)=\beta^{-1}(t)\left(-\omega^2-\epsilon\psi(t)-\dot{\alpha}(t)-\alpha^2(t)+\frac{\dot{\beta}(t)\alpha(t)}{\beta(t)}\right),$$ $C^0$-close to $\omega^2-\epsilon\psi(t)$ such that the linear differential system $$\mathcal{C}(t)=\begin{pmatrix}
\alpha(t) & \beta(t) \\ \gamma(t) & -\alpha(t)
\end{pmatrix}$$ satisfy
$$\|A(t)-\mathcal{C}(t)\|<\delta.$$
\end{enumerate}

More generally, we can define a $\delta$-neighborhood of a linear differential system with equations (\ref{ME}) on the fiber in the following way; 
\begin{definition}\label{definition}
Fix a base flow $\varphi^t\colon M\rightarrow M$ and, for each $p\in\text{Per}(\varphi^t)$ of period $\pi(p)$, take $\pi(p)$-periodic smooth coefficients $\psi_p(t)$. Fix also $\omega, \epsilon>0$. Let $A\colon \text{Per}(\varphi^t)\rightarrow \mathfrak{sl}(2,\mathbb{R})$ be the linear differential system defined by:
$$A(p,t):=\begin{pmatrix}
0 & 1 \\ -\omega^2-\epsilon\psi_p(t) & 0
         \end{pmatrix}
$$
We say that $\mathcal{U}(A,\delta)$ is a $\delta$-\emph{neighborhood of $A$} if, for each $p\in\text{Per}(\varphi^t)$,
\begin{enumerate}
\item [(i)] there exists a choice for $\dot{y}$; $\dot{y}=\alpha_p(t)y+\beta_p(t)z$ where $\alpha_p(t)$ is $C^1$-close to the zero function (w.r.t. $t$) and $\beta_p(t)$ is $C^1$-close (w.r.t. $t$) to the function $f(t)=1$;
\item [(ii)] there exists a $\pi(p)$-periodic coefficient $\gamma\colon \text{Per}(\varphi^t)\times \mathbb{R}$ defined by  $$\gamma_p(t)=\beta_p^{-1}(t)\left(-\omega^2-\epsilon\psi_p(t)-\dot{\alpha}_p(t)-\alpha_p^2(t)+\frac{\dot{\beta}_p(t)\alpha_p(t)}{\beta_p(t)}\right)$$ and $C^0$-close to $\omega^2-\epsilon\psi_p(t)$ such that the linear differential system 
\end{enumerate}
$$\mathcal{C}(p,t)=\begin{pmatrix}
\alpha_p(t) & \beta_p(t) \\ \gamma_p(t) & -\alpha_p(t)
\end{pmatrix}$$
satisfy $\|A(p,t)-\mathcal{C}(p,t)\|<\delta$, 
where $A(p,t)$ is the induced linear differential system by equation (\ref{ME}) with $\dot{y}=z$ and a $\pi(t)$-periodic coefficient $\psi_p(t)$. 
\end{definition}

\end{subsection}

\begin{subsection}{Hyperbolicity and dominated splitting}
 
Let $(M,\varphi^t,A)$ be a linear differential system. The set $\Lambda\subseteq{M}$ is said to be a \emph{uniformly hyperbolic set} if there exists a uniform $\ell\in\mathbb{N}$ such that, for any $p\in{\Lambda}$, there is a $\Phi_{A}^{t}(p)$-invariant decomposition $\mathbb{R}^{2}_p=N_{p}^{u}\oplus{N_{p}^{s}}$ satisfying the following inequalities:
$$\|\Phi_{A}^{-\ell}(p)|_{N_{p}^{u}}\|\leq\frac{1}{2} \text{  and  }\|\Phi_{A}^{\ell}(p)|_{N_{p}^{s}}\|\leq\frac{1}{2}.$$
If $\Lambda=M$, then we say that $(M,\varphi^t,A)$ is \emph{uniformly hyperbolic}. The concept of uniform hyperbolicity is equivalent to the \emph{exponential dichotomy} concept, see~\cite{C} for details. 

More generally, a $\varphi^{t}$-invariant set $\Lambda_\ell\subseteq{M}$ is said to have an  \emph{$\ell$-dominated splitting} for $(M,\varphi^t,A)$, where $p\in \Lambda_\ell$, if for any $p\in{\Lambda}_{\ell}$, there is a
$\Phi_{A}^{t}(p)$-invariant decomposition
$\mathbb{R}^{2}_p=N_{p}^{u}\oplus{N_{p}^{s}}$ satisfying,
$$\frac{\|\Phi_{A}^{\ell}(q)|_{N_{q}^{s}}\|}{\|\Phi_{A}^{\ell}(q)|_{N_{q}^{u}}\|}\leq{\frac{1}{2}}.$$ 

When the fiber $N_p^u$ dominates $N_p^s$ both fibers may contract, however $N_p^s$ is most contracting than $N_p^u$. On the other hand if both fibers expand, $N_p^s$ is less expanding than $N_p^u$ (see \cite{BDV} for more details on dominated splitting). In rough terms, in the two-dimensional Hamiltonian context, hyperbolicity is tantamount to dominated splitting. This is the content of the following lemma whose proof is elementary (see~\cite[\S 4.1]{B}).

\begin{lemma}\label{hyperbolic}
If $(M,\varphi^t,A)$ is a linear differential system, $A\in\mathfrak{sl}(2,\mathbb{R})$ and $\Lambda$ is a set with dominated splitting, then $\Lambda$ is hyperbolic. 
\end{lemma}

\end{subsection}

\begin{subsection}{An abstract general setting and statement of the result}

In order to state our results with great generality we will use the nomenclature developed in \cite{BGV} (for the discrete case) and then generalized to the continuous-time volume-preserving case in (\cite{BR}); Let $\Sigma \subset \text{Per}(\varphi^t)$ be a set formed by a countable union of closed orbits of $\varphi^t\colon M\rightarrow M$. The linear differential system $\mathcal{A}=(\Sigma, \varphi^t, A)$ is \emph{bounded} if there exists $K>0$ such that $\|A(x)\| \leq K$, for all $x \in \Sigma$. The linear differential system $\mathcal{A}$ is said to be a \emph{large period system} if the number of closed orbits of $\Sigma$ with period less or equal to $\tau$ is finite, for any $\tau >0$.

\medskip

\begin{example}
\emph{An Anosov flow (see e.g.,~\cite{BDV}) is an example of a flow with the \textbf{large period property} described above. Fix $\omega,\epsilon>0$ and take a Mathieu equation $$\ddot{y}+\left(\omega^2+\epsilon \cos\left(\frac{2\pi}{\pi(p)} t\right)\right)y=0,$$ over each closed orbit $p$ of period $\pi(p)$ of an Anosov flow to obtain $A\colon \text{Per}(\varphi^t)\rightarrow \mathfrak{sl}(2,\mathbb{R})$ which is an example of a bounded and large period system.}
\end{example}

\medskip

A linear differential system $\mathcal{B}=(\Sigma, \varphi^t, B)$  is a \emph{conservative perturbation} of a bounded system $\mathcal{A}$ if, for every $\epsilon >0$, $\|A(x) - B(x)\| < \epsilon$, up to points $x$ belonging to a finite number of orbits, and $\mathcal{B}$ is conservative i.e., if $\text{Tr}(A)=\text{Tr}(B)$.

A direct application of the Gronwall inequality gives that $$\|\Phi_{A}^t(x)-\Phi_{B}^t(x)\| \leq e^{K|t|}\|A(x) - B(x)\|.$$
In particular $\Phi_{B}^1$ is a perturbation of $\Phi_{A}^1$ in the sense introduced in~\cite{BGV} for the discrete case.

A bounded linear differential system $\mathcal{A}$ is said to be \emph{strictly without dominated decomposition} if the only invariant subsets of $\Sigma$ that admit a dominated splitting for $\Phi_{A}^t$ are finite sets.

Let us now present a key result about linear differential systems which is the conservative  flow version of~\cite[Theorem 2.2]{BGV}.

\begin{theorem}(\cite[Theorem 4.1]{BR})\label{2.2}
Let $\mathcal{A}$ be a conservative, large period and bounded linear differential system. If $\mathcal{A}$ is strictly without dominated decomposition then there exist a conservative perturbation $\mathcal{B}$ of $\mathcal{A}$ and an infinite set $\Sigma^{\prime} \subset \Sigma$ which is  $\varphi^t$-invariant such that for every $x \in \Sigma^{\prime}$ the linear map $\Phi_{B}^{\pi(x)}(x)$ as all eigenvalues real and with the same modulus (thus equal to $1$ or to $-1$).
\end{theorem}

Once we develop the perturbation framework, specific of the con\-ser\-va\-ti\-ve-flow context, previous theorem follows directly from~\cite[Theorem 2.2]{BGV}. In the present paper we can obtain a similar result, restricted to our setting, as long as we show how to perform all the perturbations in our setting. These perturbations will be the content of \S\ref{proofs}.

\medskip

Fix $\omega,\epsilon>0$, a base flow $\varphi^t\colon M\rightarrow M$ and, for each $p\in\text{Per}(\varphi^t)$, a family of $\pi(p)$-periodic smooth coefficients $\{\psi_p(t)\}$. Let \begin{equation}\label{LDS}
 \mathpzc{A}= \mathpzc{A}(p,t,\omega,\epsilon,M,\varphi^t,A,\psi_p(t))
\end{equation}
denote the \emph{Mathieu linear differential system} induced by the Mathieu equation $\ddot{y}+(\omega^2+\epsilon\psi_p(t))=0$. From now on we assume that our Mathieu linear differential systems are bounded, with large period and strictly without dominated decomposition.

We now state our main result.

\begin{maintheorem}\label{T1}
Let $\mathpzc{A}$ denote a Mathieu linear differential system as in (\ref{LDS}). Fix $\delta>0$. Then there exist $\ell\in\mathbb{N}$ and $T>0$ such that any closed orbit $p$ of period $\pi(p)\geq T$ satisfy one of the two properties:
\begin{enumerate}
\item the $\varphi^t$-orbit of $p$ is \textbf{hyperbolic}\footnote{In fact we obtain that the $\varphi^t$-orbit has an $\ell$-dominated splitting. Then, by Lemma~\ref{hyperbolic}, we obtain hyperbolicity.} with strength $\ell$ or else
\item there is a system $B\in\mathcal{U}(A,\delta)$ such that the $\varphi^t$-orbit of $p$ is \textbf{parabolic} for $B$.
\end{enumerate}

\end{maintheorem}

A few words about previous theorem. Condition (1) says that the Mathieu system with $\pi(p)$-periodic coefficient over $p$ is \emph{unstable}. Even more, the degree of instability shared by all the Mathieu equations (over points $p$ satisfying condition (1)) is uniform. In this case the characteristic multipliers are uniformly ``far'' from the tongues of instability, that is, the characteristic multipliers are uniformly away from $1$ and $-1$. Condition (2) says that, a perturbation can be made in order to the characteristic multipliers fall into the boundary of instability, this perturbation is a Hamiltonian damped Mathieu system with two large periodic coefficients. In conclusion, near any Mathieu system with parametric excitations of large period, there exists a perturbation of the two coefficients such that, if we neglect a finite number of closed orbits, any system over any closed orbit $p$ is \emph{unstable}. In other words stability tends to disappear when the period of the coefficient increases and some first order term with a small periodic coefficient is introduced.

\medskip

We end this section by noting that, Theorem~\ref{2.2} and the perturbations to be de\-ve\-lo\-ped in \S \ref{proofs}, allows us to obtain item (2) of Theorem~\ref{T1} if we have the strictly without domination property. In brief terms weak hyperbolicity can be made parabolic (see paragraph after the proof of Lemma~\ref{mpl}) and ellipticity can also be made parabolic, under certain hypothesis, by small perturbations (cf. Lemma~\ref{mpl2}).

On the other hand, if we do not have this property, we must have a uniform dominated slitting, obtaining (1).

\end{subsection}

\end{section}

\begin{section}{Construction of the perturbations}\label{proofs}

\begin{subsection}{Inducing rotations} We are going to perform some perturbations defined by  small rotations of solutions on the invariant eigendirections. Results of this type go back to Novikov ~\cite{N} and Ma\~n\'e ~\cite{M}. See also \cite{B} where perturbations of these type were made in the two-dimensional Hamiltonian setting of linear differential systems. 

\begin{lemma}\label{mpl}
Let $\mathpzc{A}$ denote the linear differential system as in (\ref{LDS}). Fix $\delta>0$. There exists $\eta>0$, such that given any $p\in\text{Per}(\varphi^t)$ (with $\pi(p)>1$), there exists $B\in\mathcal{U}(A,\delta)$ such that
\begin{enumerate}
\item [(A)] $B$ is supported in $\varphi^{t}(p)$ for $t\in[0,1]$ and
\item [(B)] $\Phi^{1}_{B}(p)=\Phi^{1}_{A}(p)\cdot{R_{\eta}}$, where $R_{\eta}$ is a rotation of angle $\eta$.
\end{enumerate}
\end{lemma}
\begin{proof}
Let $\eta\in(0,1)$ and $g\colon\mathbb{R}\rightarrow{\mathbb{R}}$ be a \emph{bump-function} defined by $g(t)=0$ for
$t<0$, $g(t)=t$ for $t\in[\eta,1-\eta]$ and $g(t)=1$ for
$t\geq{1}$. Fix a periodic point $p$ with $\pi(p)>1$ and define for all $t\geq 0$:

$$\Phi^{t}_A(p)=\Phi^t(p)=\begin{pmatrix}
a(t) & b(t) \\
c(t) & d(t) \\
\end{pmatrix}\text{  and  } R_{\eta{g(t)}}=\begin{pmatrix}
\cos(\eta{g(t)}) & -\sin(\eta{g(t)}) \\
\sin(\eta{g(t)}) & \,\,\,\,\,\cos(\eta{g(t)}) \\
\end{pmatrix}.$$

We know that $u(t)=\Phi^{t}(p)$ is a solution of the linear
variational equation $\dot{u}(t)=A(t)\cdot u(t)$. Take $\Phi^{t}(p)\cdot{R}_{\eta{g(t)}}$
and then compute the derivative w.r.t. time:
\begin{equation*}
\begin{split}
\frac{d}{dt}(\Phi^{t}(p)\cdot{R}_{\eta{g(t)}})&=\dot{\Phi}^{t}(p)R_{\eta{g(t)}}+\Phi^{t}(p)\dot{R}_{\eta{g(t)}}\\
&=A(\varphi^{t}(p))\Phi^{t}(p)R_{\eta{g(t)}}+\Phi^{t}(p)\dot{R}_{\eta{g(t)}}\\
&=A(\varphi^{t}(p))\Phi^{t}(p)R_{\eta{g(t)}}+\Phi^{t}(p)\dot{R}_{\eta{g(t)}}R_{-\eta{g(t)}}[\Phi^{t}(p)]^{-1}\Phi^{t}(p)R_{\eta{g(t)}}\\
&=[A(\varphi^{t}(p))+\Phi^{t}(p)\dot{R}_{\eta{g(t)}}R_{-\eta{g(t)}}[\Phi^{t}(p)]^{-1}]\cdot(\Phi^{t}(p)R_{\eta{g(t)}}).
\end{split}
\end{equation*}
Define $B(\varphi^{t}(p))=A(\varphi^{t}(p))+H(\varphi^{t}(p))$ where: $$H(\varphi^{t}(p))=H(\eta,t)=\Phi^{t}(p)\dot{R}_{\eta{g(t)}}R_{-\eta{g(t)}}[\Phi^{t}(p)]^{-1}.$$ 
We conclude that $v(t)=\Phi^{t}(p)\cdot{R}_{\eta{g(t)}}$ is a solution of the linear variational equation:
\begin{equation}\label{linear}
\dot{v}(s)_{|s=t}=[A(\varphi^{t}(p))+H(\eta,t)]v(t)
\end{equation}
Since,
$$\dot{R}_{\eta{g(t)}}\cdot{R}_{-\eta{g(t)}}=\eta{\dot{g}}(t)\begin{pmatrix}
0 & -1 \\
1 & 0 \\
\end{pmatrix},$$
and since $\text{Tr}(A)=0$ we use (\ref{liouville}) and we obtain that
\begin{eqnarray*}
H(\eta,t)&=&\frac{\eta{\dot{g}}(t)}{\det\Phi^{t}(p)}\begin{pmatrix}
b(t)d(t)+a(t)c(t) & -b(t)^{2}-a(t)^{2} \\
d(t)^{2}+c(t)^{2} & -b(t)d(t)-a(t)c(t) \\
\end{pmatrix}\\
&=&  \eta{\dot{g}}(t)\begin{pmatrix}
b(t)d(t)+a(t)c(t) & -b(t)^{2}-a(t)^{2} \\
d(t)^{2}+c(t)^{2} & -b(t)d(t)-a(t)c(t) \\
\end{pmatrix}.
\end{eqnarray*}

Hence $\text{Tr}(H(\eta,t))=0$, moreover since $\dot{g}(t)=0$ for
$t\notin]0,1[$, its support is $\varphi^{t}(p)$ for $t\in[0,1]$ and (A) is proved.
Since $t\in[0,1]$ and all the terms in the definition of
$H(\eta,t)$ are uniformly bounded for all $p\in{M}$ (by definition of these type of systems), given any size
of perturbation allowed by $\delta>0$ we take $\eta(\delta)$ sufficiently
small to guarantee that we get $\|A-B\|=\|H\|<\delta$.

Finally, for (B), we observe that $v(t)=\Phi_{A}^{t}(p)\cdot{R}_{\eta{g}(t)}$ is solution of~(\ref{linear}). So for $t=1$ we obtain $\Phi_{B}^{1}(p)=\Phi_{A}^{1}(p)\cdot{R}_ {\eta}$. 

To conclude the proof of the lemma we only need to see that $B$ generates a traceless damped Mathieu system with two perturbed coefficients $\Gamma(t)$ and $\Psi(t)$ close to zero and $\psi(t)$ respectively.
\begin{eqnarray*}
B&=&A(t)+H(\eta,t)\\
&=&\begin{pmatrix}
0 & 1 \\ -\omega^2-\epsilon\psi(t) & 0
\end{pmatrix}+\eta{\dot{g}}(t)\begin{pmatrix}
b(t)d(t)+a(t)c(t) & -b(t)^{2}-a(t)^{2} \\
d(t)^{2}+c(t)^{2} & -b(t)d(t)-a(t)c(t) \\
\end{pmatrix}\\
&=&\begin{pmatrix}
\eta{\dot{g}}(t)(b(t)d(t)+a(t)c(t)) & 1-\eta{\dot{g}}(t)(b(t)^{2}+a(t)^{2}) \\ -\omega^2-\epsilon\psi(t)+\eta{\dot{g}}(t)(d(t)^{2}+c(t)^{2}) & -\eta{\dot{g}}(t)(b(t)d(t)+a(t)c(t))
\end{pmatrix}.
\end{eqnarray*}

We abbreviate $F=b(t)d(t)+a(t)c(t)$, $G=b(t)^{2}+a(t)^{2}$ and $H=d(t)^{2}+c(t)^{2}$ and thus,

$$B=\begin{pmatrix}
\eta{\dot{g}}(t)F & 1-\eta{\dot{g}}(t)G \\ -\omega^2-\epsilon\psi(t)+\eta{\dot{g}}(t)H & -\eta{\dot{g}}(t)F
\end{pmatrix}.$$

Using equation (\ref{form}) we obtain (denoting also by $^\prime$ the derivative w.r.t. $t$):

$$
\ddot{y}-\left(\frac{(1-\eta\dot{g}G)^{\prime}}{1-\eta\dot{g}G}\right)\dot{y}+\left(\det(B)-(\eta\dot{g}F)^{\prime}+\frac{(1-\eta\dot{g}G)^{\prime}\eta\dot{g}F}{1-\eta\dot{g}G}\right)y=0,
$$
which is equivalent to
$$
\ddot{y}+\left(\frac{\eta\ddot{g}G+\eta\dot{g}\dot{G}}{1-\eta\dot{g}G}\right)\dot{y}+\left(\det(B)-\eta\ddot{g}F-\eta\dot{g}\dot{F}+\frac{(-\eta\ddot{g}G-\eta\dot{g}\dot{G})\eta\dot{g}F}{1-\eta\dot{g}G}\right)y=0.
$$
Now, since $$\det(B)=\omega^2+\epsilon\psi(t)+\eta^2\dot{g}^2(GH-F^2)-\eta\dot{g}(H+\omega^2 G+G\epsilon \psi(t)),$$
we define the $\tau$-periodic coefficients by:
$$
\Gamma_\eta(t):=\frac{\eta\ddot{g}G+\eta\dot{g}\dot{G}}{1-\eta\dot{g}G}$$

and 

$$\Theta_{\eta}(t):=- \eta\ddot{g}F-\frac{(\ddot{g}G+\dot{g}\dot{G})\eta^2\dot{g}F}{1-\eta\dot{g}G}+\eta^2\dot{g}^2(GH-F^2)-\eta\dot{g}(\dot{F}+H+\omega^2 G+G\epsilon \psi(t)).
$$
Finally, we obtain,

\begin{equation}
 \ddot{y}+\Gamma_\eta(t)\dot{y}+(\omega^2+\epsilon\Psi_\eta(t))=0,
\end{equation}
where $\Psi_\eta(t):=\psi(t)+\epsilon^{-1}\Theta_{\eta}(t)$. Observe that, if $\eta$ is taken very small, then $\Gamma_\eta(t)$ is close to zero, moreover, since $\Theta_\eta(t)\underset{\eta\rightarrow 0}{\rightarrow} 0$,  $\Psi_\eta(t)$ is close to $\psi(t)$. Then $B\in\mathcal{U}(A,\delta)$ according to Definition~\ref{definition} and the lemma is proved.

\end{proof}

We observe that condition (2) of Theorem~\ref{T1} assures that there is a Hamiltonian damped Mathieu system $B$ with perturbed coefficients, arbitrarily close to the original one, such that $\Phi_B^{\pi(p)}(p)$ has only real eigenvalues of the same modulus (thus equal to one). A perturbation similar to the one constructed in Lemma~\ref{mpl} can be done, by adding small directional homotheties, in order to obtain two real characteristic multipliers (hyperbolic case). 

\medskip

The next crucial lemma show how to create instability (parabolic or hyperbolic points), by small perturbations, once we are in the presence of an elliptic point and we have enough ``time'' to perturb.

\begin{lemma}\label{mpl2}
Let $\mathpzc{A}$ denote the Mathieu linear differential system as in (\ref{LDS}) and $\delta>0$. Then there exists $T>0$ such that for any closed orbit $p$ of period $\pi(p)\geq T$ there is $B\in\mathcal{U}(A,\delta)$ satisfying
\begin{itemize}
\item all the eigenvalues of $\Phi_B^{\pi(p)}$ are real, and 
\item $B=A$ outside a small neigh\-bour\-hood of the orbit of $p$.  
\end{itemize}
\end{lemma}
\begin{proof}

Fix a small $\delta >0$. Let $R_{\theta}$ denote the plane rotation of angle $\theta$. By ~\cite[Lemma 6.6]{BC} there exists $N=N(\epsilon) \in \mathbb{N}$ satisfying the following:  for any $k>N$ and for any $C_1$, $C_2$,...,$C_k \in sl(2,\mathbb{R})$ there are rotations $R_{\theta_1}$,   $R_{\theta_2}$,...,$R_{\theta_k}$, with $|\theta_j|< \delta$ for all $j \in \{1,2,...,k\}$, such that the linear map 
$$ C_k \cdot R_{\theta_k} \cdot C_{k-1} \cdot R_{\theta_{k-1}} \cdot ... \cdot C_1 \cdot R_{\theta_1}$$ has real eigenvalues.

Let us fix a periodic orbit $\Omega$ and $p \in \Omega$ with $\pi(p) \geq N$. We assume that $\Phi_A^{\pi(p)}(p)$ has a complex eigenvalue. Assuming that $\pi(p)=k \in \mathbb{N}$, we consider, for each $j\in\{1,...,k\}$ the linear maps $C_j \colon \mathbb{R}^2 \rightarrow \mathbb{R}^2$ defined by 
$$C_j=\Phi_A^1(\varphi^{j-1}(p)).$$ 
If $\pi(p) \notin \mathbb{N}$ we take $k=[\pi(p)]$, consider $C_1,...,C_{k-1}$ as before and define 
$$
\begin{array}{cccc}
C_k\colon & \mathbb{R}^2_{\varphi^{k-1}(p)} & \longrightarrow & \mathbb{R}^2_{\varphi^{\pi(p)}(p)} \\
& v & \longmapsto & \Phi_A^{1+\pi(p)-k}(\varphi^{k-1}(p))\cdot \, v.
\end{array}
$$

In what follows, without loss of generality, we assume that $\pi(p)=k \in \mathbb{N}$. 

We observe that each $C_j$ can be identified with a linear map of $sl(2, \mathbb{R})$ and that $\Phi_A^{\pi(p)}(p)=C_k \cdot C_{k-1}\cdot(...)\cdot C_1$. Therefore, \cite[Lemma 6.6]{BC} gives a family of  rotations $R_{\theta_1}$,   $R_{\theta_2}$,...,$R_{\theta_k}$ with the properties described above. 

Now we apply Lemma~\ref{mpl} to each arc $\{\varphi^{j-1+t}(p)\colon t\in[0,1]\}$ and to the maps $C_j$ and $R_{\theta_j}$ and realize the perturbations, i.e., find adequate coefficients $\Gamma_p^{j}(t)$ and $\Psi_p^j(t)$. 

\end{proof}

\end{subsection}

\end{section}

\section*{Acknowledgements}

The author was partially supported by the FCT-Funda\c{c}\~ao para a Ci\^encia e a Tecnologia -- project PTDC/MAT/099493/2008 and also SFRH/BPD/20890/2004.



\begin{thebibliography}{ABC}



\bibitem{B} M. Bessa, \emph{Dynamics of generic 2-dimensional linear differential systems}, Jr. Diff. Eq., 228, 2, (2006), 685--706.



\bibitem{BR} M. Bessa and J. Rocha, \emph{On $C^1$-robust transitivity of volume-preserving flows}, Jr. Diff. Eq., 245, 11, (2008), 3127--3143.


\bibitem{BC} C. Bonatti and S. Crovisier, \emph{R\'ecurrence et g\'en\'ericit\'e}, Invent. Math. 158, 1 (2004), 33--104.

\bibitem{BDV}
C.~Bonatti, L.~J. D{\'{i}}az, and M.~Viana.
\newblock Dynamics beyond uniform hyperbolicity, {102} of {\em
  {Encyclopaedia of Mathematical Sciences}}.
\newblock {Springer-Verlag}, {Berlin}, {2005}.
\newblock {A global geometric and probabilistic perspective, Mathematical
  Physics, III}.

\bibitem{BGV} C. Bonatti, N. Gourmelon and T. Vivier, \emph{Perturbations of the derivative along periodic orbits.}
Ergod. Th. \&\ Dynam. Sys. 26, 5 (2006), 1307--1337.

\bibitem{HK}
J. Hale and H. Ko\c{c}ak,
\newblock {\em {Dynamics and bifurcations}}, {3}, {\em {Texts in Applied Mathematics}}.
\newblock {Springer-Verlag}, {New York}, {1991}.


\bibitem{C} W. Coppel, Dichotomies in Stability Theory. Lecture Notes in Mathematics 629. Springer Verlag, 1978.


\bibitem{M}
R.~Ma{\~{n}}\'{e},
\newblock \emph{An ergodic closing lemma},
\newblock Annals of Math., \textbf{116}, (1982), 503--540.

\bibitem{Ma}
E.~Mathieu,
\newblock \emph{M\'emoire sur Le Mouvement Vibratoire d’une Membrane de forme Elliptique},
\newblock {\em Journal des Math\'ematiques Pures et Appliqu\'es}, (1868), 137--203.


\bibitem{N} V. L. Novikov, \emph{Almost reducible systems with almost periodic coefficients,}
Mat. Zametki 16 (1974), 789--799.


\bibitem{TN} J. H. Taylor and K. Narendra, \emph{Stability regions for the damped Mathieu equation,}
SIAM J. Appl. Math.  17,  (1969) 343--352.


\end{thebibliography}
\end{document}